\documentclass[journal]{IEEEtran}

\usepackage{amsmath,amssymb,amsfonts,amsthm}
\usepackage{mathtools}
\usepackage{bm}
\usepackage{graphicx}
\usepackage{booktabs}
\usepackage[table]{xcolor}
\usepackage{colortbl}
\usepackage{algorithm}
\usepackage{algpseudocode}
\usepackage{dsfont}
\usepackage{microtype}
\usepackage{hyperref}
\usepackage{tikz}
\usetikzlibrary{positioning,calc}

\usepackage{xcolor}

\usepackage{url}
\usepackage{hyperref}
\hypersetup{colorlinks,allcolors=black}

\newtheorem{lemma}{Lemma}

\newtheorem{example}{Example}

\newcommand{\ie}{{\it i.e.}}
\newcommand{\eg}{{\it e.g.}}

\begin{document}


\title{Learning to Decode Quantum LDPC Codes Via Belief Propagation}

\author{
Mohsen~Moradi\textsuperscript{\href{https://orcid.org/0000-0001-7026-0682}{\includegraphics[scale=0.06]{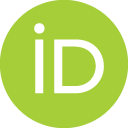}}}, 
Vahid~Nourozi\textsuperscript{\href{https://orcid.org/0000-0002-6977-865X}{ \includegraphics[scale=0.06]{figs/ORCID}}} ,
Salman~Habib\textsuperscript{\href{https://orcid.org/0000-0003-0699-5312}{\includegraphics[scale=0.06]{figs/ORCID}}}
and 
David~G.~M.~Mitchell\textsuperscript{\href{https://orcid.org/0000-0002-3544-9225}{ \includegraphics[scale=0.06]{figs/ORCID}}}%
\thanks{Mohsen Moradi, Vahid Nourozi, and David G. M. Mitchell are with the Klipsch School of Electrical and Computer Engineering, New Mexico State University, Las Cruces, NM 88003, USA (e-mail: moradi23@nmsu.edu; nourozi@nmsu.edu; dgmm@nmsu.edu).}%
\thanks{Salman Habib is with the Dept. of Electrical and Computer Engineering, Texas A \& M University, Texarkana, TX 75503, USA (e-mail: shabib@tamut.edu).}%
}

\maketitle

\begin{abstract}
Belief-propagation (BP) decoding for quantum low-density parity-check (QLDPC) codes is appealing due to its low complexity, yet it often exhibits convergence issues due to quantum degeneracy and short cycles that exist in the Tanner graph.
To overcome this challenge, this paper proposes a reinforcement-learning (RL) approach that learns (offline) how to decode QLDPC codes based on sequential decoding trajectories. The decoding is formulated as a Markov decision process with a local, syndrome-driven state representation of the underlying RL agent. To enable fast inference, critical for practical implementation, we incrementally update our RL-based QLDPC decoder using second-order neighborhoods that avoid global rescans. 
Simulation results on representative QLDPC codes demonstrate the superiority of the proposed RL-based QLDPC decoders in terms of performance and convergence speed when compared to flooding and random sequential schedules, while achieving performance competitive with state-of-the-art BP-based decoders at comparable complexity. 
\end{abstract}


\section{Introduction}
Quantum error correction is essential for reliable quantum computing, since even small physical error rates can accumulate and corrupt computations without mitigation \cite{kuo2022exploiting, shor1995scheme, steane1996error}. 
Quantum low-density parity-check (QLDPC) codes have emerged as a promising class of codes for fault tolerance because they can protect multiple logical qubits with sparse stabilizer constraints. 
In particular, recent advances in QLDPC constructions 
have produced codes with constant rate and linear growth of minimum distance, making them increasingly attractive for scalable quantum architectures \cite{tillich2013quantum, panteleev2022asymptotically, breuckmann2021balanced, bravyi2024high}.
However, the decoding of QLDPC codes poses significant challenges. Standard belief propagation (BP) decoding of classical LDPC codes does not result in good performance for QLDPC decoding due to the presence of numerous short cycles, which violate BP’s independence assumptions, and the phenomenon of {\it degeneracy} \cite{kuo2022exploiting}, where multiple distinct error patterns yield the same syndrome, creating symmetric \emph{pseudo-codewords} that can confuse the decoder. As a result, applying a conventional BP decoder for QLDPC codes typically fails to converge, frequently oscillating or getting trapped in a wrong error coset, especially at low error rates. 

Techniques attempting to maintain the low complexity of BP decoding and achieve improved performance for QLDPC codes has drawn a lot of attention \cite{poulin2008iterative, fuentes2021degeneracy,roffe2020decoding, miao2025quaternary}; in particular, techniques to address convergence failures. Hybrid post-processing decoders augment BP with a secondary decoding stage to correct its mistakes. For example, BP with ordered statistics decoding (BP-OSD) \cite{panteleev2021degenerate} runs BP first, then applies an OSD that solves a system of linear equations to find an error consistent with the syndrome, improving the convergence probability. Various modifications to this have been proposed, see, e.g., \cite{roffe2023bias}. Likewise, BP with stabilizer inactivation (BP-SI) temporarily disables or \emph{inactivates} certain check nodes (CNs) during BP when they are suspected of causing trapping sets, thereby breaking the symmetry caused by degenerate errors \cite{du2022stabilizer}. These hybrid strategies can close much of the performance gap of plain BP, but they come at the cost of substantially higher computation (e.g. Gaussian elimination for OSD). 
{More recently, \cite{hillmann2025localized} proposes BP with localized statistics decoding (BP-LSD), which achieves performance comparable to BP-OSD while reducing post-processing cost: instead of performing a single global elimination over the entire decoding graph, LSD restricts the linear solves to multiple small, independently processed subgraphs (which can also be parallelized)}.




Researchers have also investigated modifying the BP process itself. Overcomplete-check BP and its quaternary neural extension provide a low-latency, high-performance decoding framework \cite{miao2023quaternary}. 
Another work that uses the BP message-passing process itself to improve convergence is BP guided decimation (BPGD) \cite{yao2023bpgd}. BPGD interleaves fixed runs of BP with decimation steps that \emph{freeze} the most reliable variable nodes (VNs) (\ie, fix certain qubit error values to 0 or 1 based on high-confidence beliefs) to inject asymmetry into the decoding. By sequentially fixing VNs, BPGD shrinks the solution space and avoids the oscillations of loopy BP, reducing the failure rate due to non-convergence. This decimation approach can achieve error-rate performance on par with BP-OSD and BP-SI, without the need for expensive linear-algebra post-processing. Variants of this idea incorporating soft reliability thresholds and partial decimation have further improved BP’s efficacy at the cost of some complexity \cite{alinia2025decimation}. 

In~\cite{moradi2026seq_qldpc}, flooding BP is replaced by sequential BP with random scheduling, which was shown to further improve the decoder performance in terms of both convergence speed and error performance. These results suggest that injecting sequential asymmetry into BP (either by decimation or scheduling) is a powerful strategy to combat the traps caused by degeneracy. In classical coding theory, it is well known that the schedule of message updates in BP can influence its performance. Instead of the conventional flooding schedule (updating all nodes in parallel each round), one can use a sequential (serial) schedule that updates one node at a time in some order. Even a simple randomized sequential update (often called shuffled BP) can mitigate the adverse effects of short cycles and improve convergence for classical LDPC decoders \cite{zhang2005shuffled}. Researchers have subsequently cast schedule optimization as a learning problem, using reinforcement learning (RL) agents to discover effective update sequences automatically \cite{ carpi2019reinforcement, habib2021belief, moradi2025enhancing}.

Random scheduling has been explored for decoding QLDPC codes in \cite{moradi2026seq_qldpc}. It was demonstrated there that processing QLDPC VNs or CNs in a fixed sequential order can stabilize BP and reduce stall patterns. For instance, a sequential BP decoder that updates each CN or VN one-by-one (instead of all-at-once) achieves a lower frame error rate (FER) than standard flooding BP on several QLDPC benchmarks for a moderate cap on the number of iterations. In fact, a sequential update strategy is shown to outperform a BP-OSD-0 decoder for some QLDPC codes, at roughly the same average decoding complexity as BP decoding. These findings underscore that the order of message updates is crucial for breaking the symmetric failures of BP on QLDPC codes. 

In this paper, we propose an RL framework for \emph{learning} effective sequential
schedules for BP-based decoding of QLDPC codes. Our main contributions are as follows:

\begin{itemize}
    \item We first establish the RL-based decoding approach in the context of Calderbank-Shor-Steane (CSS) codes for an independent Pauli
    $X$ channel by casting the decoder scheduling as a Markov decision process with a \emph{local, syndrome-driven} state, where
    for each candidate VN, the state is determined by the residual mismatch pattern on its adjacent CNs,
    which remains small because the Tanner graph is sparse. A Q-learning agent then selects VN
    updates \emph{without replacement} within each BP iteration, enabling the decoder to inject controlled asymmetry
    purely through the update order, while maintaining the underlying sequential VN scheduling (SVNS) message-update rule; 
    \item To make the learned
    policy practical at block-lengths of interest, we develop a fast inference implementation that incrementally maintains
    the residual syndrome, local states, and scheduling priorities using only the \emph{second-order neighborhood} affected by
    a hard-decision flip, avoiding global rescans and reducing overhead; 
    \item We extend our proposed learning-based scheduling to depolarizing noise by using a two-stream (quaternary-coupled) SVNS update and quaternary hard decisions. 
\end{itemize}



\noindent Numerical results demonstrate the superiority of
the proposed RL-based QLDPC decoders in terms of performance
and convergence speed when compared to flooding and random
sequential schedules, while achieving performance competitive
with state-of-the-art BP-based decoders at comparable complexity. Finally, we note that our method is modular, \ie, the proposed learned sequential schedules can be combined with other techniques providing a further way to improve both convergence behavior and end-to-end decoding performance.

The remainder of this paper is organized as follows.
Section~\ref{sec:background} reviews the necessary background on stabilizer and CSS QLDPC codes, syndrome-based decoding,
and BP/SVNS message passing.
In Section~\ref{sec:RL_SVNS}, we introduce our RL formulation for state-dependent SVNS and establish the required notation.
Section~\ref{sec:fast} presents an efficient implementation based on incremental/local updates to reduce inference-time
complexity.
Section~\ref{sec:rl_svns_depolarizing} extends the RL scheduling framework to the depolarizing channel and develops the corresponding notation and update rules.
Section~\ref{sec:numerical} reports numerical results on several standard QLDPC codes and compares our methods with state-of-the-art decoders.
Finally, Section~\ref{sec:conclusion} concludes the paper.

\section{Background}\label{sec:background}

\subsection{Stabilizer and CSS QLDPC Codes}
Quantum stabilizer codes are the quantum analog of classical linear block codes. They are specified by a set of commuting
Pauli operators (stabilizer generators) whose joint $+1$ eigenspace defines the codespace \cite{kuo2022exploiting}.
Each stabilizer acts as a parity-check measurement on a subset of qubits, and commutativity ensures that all stabilizers
can be measured simultaneously without disturbing the encoded logical information.

A stabilizer code with $m$ generators can be represented by a binary \emph{symplectic} matrix
$H=[\,H_X\;H_Z\,]\in\{0,1\}^{m\times 2n}$, where the $i$-th row specifies the $X$-support and $Z$-support
of the $i$-th stabilizer generator.
The commutativity of all stabilizers is equivalent to the symplectic orthogonality condition
\[
H_XH_Z^\top + H_ZH_X^\top = 0 \quad \text{over } \mathbb{F}_2.
\]
In this paper, we restrict attention to CSS codes \cite{calderbank1996good}, a widely used subclass of stabilizer codes. A CSS code splits into $X$-type and $Z$-type checks. In the CSS
construction adopted here, we use
\[
H_X=\begin{bmatrix}0\\ G_2\end{bmatrix},\qquad
H_Z=\begin{bmatrix}H_1\\ 0\end{bmatrix},
\]
where $H_1\in\{0,1\}^{(n-k_1)\times n}$ is the parity-check matrix of a classical code $C_1$
and $G_2\in\{0,1\}^{k_2\times n}$ is the generator matrix of another classical code $C_2$.
The commutativity condition reduces to $G_2H_1^\top=0$ (equivalently $C_2\subseteq C_1$).

\subsection{Syndrome decoding and degeneracy (X-only channel)}
In this paper we first focus on the independent Pauli-$X$ error channel for conceptual clarity.
In this setting the physical error is a binary vector $\bm{e}\in\{0,1\}^n$ (bit flips on $n$ qubits).
The relevant syndrome component is
\[
\bm{s} \triangleq H_1 \bm{e} \in\{0,1\}^{m_1},\qquad m_1=n-k_1,
\]
A decoder outputs an estimate $\hat{\bm{e}}$ such that
\[
H_1 \hat{\bm{e}} \equiv \bm{s} \pmod 2.
\]

A key difference from classical decoding is \emph{degeneracy}, where multiple distinct physical error patterns can produce the
same syndrome when they differ by stabilizers \cite{kuo2022exploiting}. In the present $X$-only CSS setting, if two error
vectors differ by a codeword in $\ker(H_1)$, they are syndrome-equivalent and correspond to the same error coset. This
degeneracy can create many competing solutions for the syndrome with similar likelihood, which complicates BP and can lead to decoding failure.

Importantly, matching the measured syndrome is not sufficient for successful quantum decoding.
Even if the decoder outputs an estimate $\hat{\bm{e}}$ such that $H_1\hat{\bm{e}}=\bm{s}$, the correction can still induce a
\emph{logical error} if $\hat{\bm{e}}$ lies in a different logical coset than the true error $\bm{e}$ (\ie, if the residual
difference $\hat{\bm{e}}\oplus\bm{e}$ contains a nontrivial logical component rather than being generated solely by stabilizers).
In contrast, we call it a \emph{decoder failure} when the algorithm terminates without producing any estimate $\hat{\bm{e}}$
that satisfies $H_1\hat{\bm{e}}=\bm{s}$ within its iteration or time budget.
The main performance metric is the frame error rate (FER), defined as the probability that decoding is not successful, i.e.,
\[
\text{FER}\;\triangleq\;\Pr\!\big(\text{logical error}\ \text{or decoder failure}\big).
\]

\subsection{Residual Mismatch}
We define the residual mismatch vector
\[
\bm{\delta} \triangleq \bm{s} \oplus (H_1\hat{\bm{e}})\in\{0,1\}^{m_1}.
\]
We say that CN $j$ is {\it satisfied} if $\delta_j=0$ and unsatisfied if $\delta_j=1$. We define the mismatch weight
\[
w\triangleq\|\bm{\delta}\|_1=\sum_{j=1}^{m_1}\delta_j.
\]
This residual view is convenient for sequential decoding and will also serve as the basis for our RL state design.

\subsection{Sequential BP Decoding}\label{subsec:SVNS}
{We represent $H_1$ with a bipartite Tanner graph $G_{H_1}=(V\cup C,E)$ with
$
V=\{v_1,\ldots,v_n\}$ and $ C=\{c_1,\ldots,c_m\},$
where each VN $v_i$ corresponds to qubit $i$ (and the error $e_i$), and each check node $c_j$ corresponds to the $j$th parity-check equation (and syndrome bit $s_j$). An edge connects $v_i$ and $c_j$ precisely when qubit $i$ participates in check $j$, i.e.,
$E=\{(v_i,c_j): H_{1}(j,i)=1\}.$
We use $\mathcal{N}(c_j)=\{i\;|\;(v_i,c_j)\in E\}$ to denote the VN neighbors of check $c_j$
and $\mathcal{N}(v_i)=\{j\;|\;(v_i,c_j)\in E\}$ to denote the CN neighbors of variable $v_i$.}
BP is an iterative message-passing heuristic that operates on $G_{H_1}$.
Under an independent and identically distributed (i.i.d.)\ bit-flip model with flip probability $p_x$ (corresponding to i.i.d. Pauli $X$ errors), each VN is initialized with the prior LLR
\[
\mu_x=\log\!\Big(\frac{1-p_x}{p_x}\Big).
\]
BP then iteratively exchanges messages between checks and variables to refine beliefs about $\bm{e}$.

In standard \emph{flooding} BP, all CN updates are computed in parallel and then all VN updates are computed in parallel
within each iteration. In contrast, \emph{sequential} (serial) schedules update one node at a time and immediately use the
new information in subsequent updates within the same iteration. Such schedules can break symmetries, propagate reliable information
faster, and sometimes mitigate trapping/oscillation behavior.

In this work we use {\it sequential variable node scheduling} (SVNS): instead of updating all VNs in parallel, we pick one
VN $v_i$ at a time and apply one local BP-style update as in \cite{moradi2026seq_qldpc}. We maintain variable beliefs as log-likelihood ratios (LLRs) $L_i$ and edge messages
$m_{i\to j}$ on each Tanner-graph edge $(v_i, c_j)$. 
A single SVNS update at a VN $v_i$ has the following steps:
\begin{enumerate}
\item Compute check-to-variable messages $m_{j\to i}$ for each neighbor $j\in\mathcal{N}(v_i)$ using the standard BP $\mathrm{tanh}$-product rule
with the syndrome sign $(-1)^{s_j}$ as
\begin{equation}\label{eq:c2v}
m_{j\to i}
=
2\,\mathrm{atanh}\!\Bigg(
(-1)^{s_j}\prod_{i'\in\mathcal{N}(c_j)\setminus\{i\}} \tanh\!\big(\tfrac{m_{i'\to j}}{2}\big)
\Bigg);
\end{equation}

\item Update the belief at VN $v_i$ as
\[
L_i \leftarrow \mu_x + \sum_{j\in\mathcal{N}(v_i)} m_{j\to i};
\]

\item Update outgoing VN to CN messages as
\[
m_{i\to j} \leftarrow L_i - m_{j\to i},\qquad \forall j\in\mathcal{N}(v_i);
\]

\item Update the hard decision as
\[
\hat e_i \leftarrow \mathds{1}\!\left[L_i<0\right]\in\{0,1\};
\]
\end{enumerate}
We refer the reader to \cite{moradi2026seq_qldpc} for a detailed review of the SVNS decoding algorithm.

\subsection{Motivation for RL-based Scheduling}
Although sequential schedules can improve BP performance, the best update order depends strongly on the specific syndrome instance
and the evolving decoder state. Fixed schedules (deterministic or random) may therefore be suboptimal across different errors.
This motivates learning a \emph{state-dependent} scheduling policy. In the next sections, we formulate VN scheduling as an RL problem and learn which VN to update next based on local, syndrome-driven features, while keeping the SVNS
message-update rule itself unchanged.

\section{RL Formulation for Scheduling VN Updates}\label{sec:RL_SVNS}
We keep the SVNS update rule fixed and learn \emph{which VN to update next}.
The following lemma shows that if one bit of the current estimated error vector flips during an SVNS update, then the residual mismatch bits of the CNs connected to that VN also flip.

\begin{lemma}\label{lem:delta_flip}
For a measured syndrome $\bm{s}$, and the current hard estimate $\hat{\bm{e}}$ with the residual mismatch vector
$
\bm{\delta} = \bm{s} \oplus (H_1 \hat{\bm{e}}),
$
if we flip the $i$-th bit of the hard estimate $\hat{\bm{e}}$
then the updated residual mismatch bits flip exactly on the neighboring checks of VN~$i$.
\end{lemma}

\begin{proof}
Let $\mathbf{e}_i=(0,0,\ldots,e_i=1,0,\ldots,0)^\top\in\{0,1\}^n$
denote the $i$-th unit vector and define the flipped hard estimate
\[
\hat{\bm{e}}' \triangleq \hat{\bm{e}} \oplus \mathbf{e}_i .
\]
By linearity over $\mathbb{F}_2$,
\[
H_1 \hat{\bm{e}}' \;=\; H_1(\hat{\bm{e}} \oplus \mathbf{e}_i)
\;=\; (H_1\hat{\bm{e}})\oplus(H_1\mathbf{e}_i).
\]
The updated residual mismatch vector is
\[
\bm{\delta}' \triangleq \bm{s} \oplus (H_1\hat{\bm{e}}')
= \bm{s} \oplus (H_1\hat{\bm{e}})\oplus(H_1\mathbf{e}_i)
= \bm{\delta} \oplus (H_1\mathbf{e}_i).
\]
Note that $H_1\mathbf{e}_i$ selects the $i$-th column of $H_1$. Therefore, for each
check index $j$,
\[
\delta'_j = \delta_j \oplus H_1(j,i),
\]
so $\delta'_j$ flips if and only if $H_1(j,i)=1$, i.e., exactly for the neighboring checks of VN~$v_i$, \ie, for $j \in \mathcal{N}(v_i)$.
\end{proof}

We use this lemma to show that the updates under the learned sequential schedule are incremental and local.
The components of our RL scheme are as follows:
\begin{itemize}
\item \textbf{Episode:} sample $p_x$ uniformly from a training grid $\mathcal P_{\text{train}}$ (in the numerical results we use $P_{\text{train}} = \{0.03, 0.04, 0.05, 0.06, 0.07\}$), draw $\bm{e}\sim\text{Bernoulli}(p_x)$, compute $\bm{s}=H_1\bm{e}$, and run SVNS with RL scheduling for up to $I_{\max}$ iterations;

\item \textbf{Environment:} the full decoder configuration: $\{m_{i\to j}\}$, $\{L_i\}$, $\hat{\bm{e}}$, and $\bm{\delta}$;

\item \textbf{State:} for each VN $v_i$, the local state $\sigma_i$ is computed from $\bm{\delta}$ restricted to the index set $\mathcal{N}(v_i)$;

\item \textbf{Agent Action:} choose which VN index $i$ to update next. Inside each BP iteration, schedule the VNs \emph{without replacement}, \ie, each VN is updated at most once per iteration maintaining a \emph{remaining set}.\footnote{We have also run learning with no restriction from the remaining set which produced approximately equal numerical results.}
\end{itemize}

For each VN $v_i$, let $A_i=|\mathcal{N}(v_i)|$ and we define $A_{\max}\triangleq \max_i A_i$.
We build a local binary state vector from residual mismatch bits of its neighbor checks as
\[
\bm{b}_i=[\,\delta_{j_1},\,\delta_{j_2},\,\ldots,\,\delta_{j_{A_i}},\,0,\ldots,0\,]\in\{0,1\}^{A_{\max}},
\]
where $(j_1,\dots,j_{A_i})$ is a fixed deterministic neighbor order (e.g., sorted check indices).
To define the local state $\sigma_i$ for the VN $v_i$, we map $\bm{b}_i$ into an integer using its binary representation as
\[
\sigma_i\triangleq \sum_{k=1}^{A_{\max}} b_{i,k}\,2^{k-1}\in\{0,\dots,2^{A_{\max}}-1\},\qquad
S_{\max}\triangleq 2^{A_{\max}}.
\]

We use a state-per-action table to obtain the Q-table $Q\in\mathbb{R}^{S_{\max}\times n}$, where the element $Q(\sigma,i)$ is the value of scheduling VN $v_i$ when its local state is $\sigma$.
At each scheduling step $t$, for the action $a$ we define
\[
w_{\text{before}}\triangleq \|\bm{\delta}_{\text{before}}\|_1,\qquad
w_{\text{after}}\triangleq \|\bm{\delta}_{\text{after}}\|_1,
\]
where $\bm{\delta}_{\text{before}}$ and $\bm{\delta}_{\text{after}}$ denote the residual mismatch vectors immediately before and after applying one SVNS update at VN~$v_i$, respectively. We then define the reward as
\[
r_t\triangleq \frac{w_{\text{before}}-w_{\text{after}}}{A_a},\qquad A_a=|\mathcal{N}(v_a)|.
\]
We also add a terminal bonus of $+1$ to the reward when $w_{\text{after}}=0$.
After action $a$, the one-step lookahead is
\[
\text{best\_future}\triangleq  \max_{i'\in\text{remaining}} Q(\sigma_{i'},i').
\]

Algorithm~\ref{alg:RL_SVNS} summarizes the proposed Q-learning-based training procedure for learning the VN scheduling policy used in our decoder.
In line~1, we initialize all Q-table entries to zero.
We train for $E_{\max}$ episodes; in our numerical results we use $E_{\max}=10^5$.

In each episode (lines~3-6), we first sample a training crossover probability $p_x$ from $\mathcal{P}_{\text{train}}$,
generate an $X$-error pattern and its syndrome, and initialize the BP beliefs/messages.
We then form an initial hard decision and the corresponding residual mismatch.
These steps set the initial decoding state for the learning procedure.

Next, in lines~7-29, we run the learning SVNS procedure for at most $I_{\max}$ BP iterations
(we use $I_{\max}=100$ in our numerical results).
At the beginning of each BP iteration (line~8), we define \emph{remaining} as the set of VNs that have not been
visited in the current SVNS iteration, and we schedule nodes \emph{without replacement} by removing the selected node from
\emph{remaining} (line~19).
Lines~9-28 iterate over the current \emph{remaining} set.
Line~10 checks if the residual mismatch is zero and terminates early if the decoder has converged.

At each step, we compute the state $\sigma_i$ for each candidate node in \emph{remaining} (line~13) and select an action using
the $\varepsilon$-greedy rule (line~14).\footnote{With probability $\varepsilon$ we explore by choosing a VN uniformly at random from the current \emph{remaining} set,
and with probability $1-\varepsilon$ we exploit the learned scores by choosing
\[
a \in \arg\max_{i\in\text{remaining}} Q(\sigma_i,i),
\]
breaking ties uniformly at random.}
We then apply one SVNS update at the chosen node and compute the corresponding reward based on the change in mismatch
(lines~15-18).
Lines~20-25 compute the one-step lookahead term by checking whether \emph{remaining} is empty. If not, we evaluate the best
future value over the remaining actions.
Finally, lines~26-27 update the Q-table entry corresponding to the chosen action.

After training, we load policy  $Q$ and run the same SVNS decoder but with $\varepsilon=0$ and no Q updates to have action
\[
a=\arg\max_{i\in\text{remaining}} Q(\sigma_i,i).
\]

In the next section, we develop local and incremental updates that significantly reduce the per-step complexity of our learned-based SVNS BP decoder while preserving its decoding performance.

\begin{algorithm}[t]
\caption{Q-learning for SVNS VN scheduling }\label{alg:RL_SVNS}
\begin{algorithmic}[1]
\State Initialize $Q\leftarrow 0$
\For{episode $=1$ to $E_{\max}$}
  \State Sample $p_x$ uniformly from $\mathcal P_{\text{train}}$
  \State Sample $\bm{e}\sim\text{Bernoulli}(p_x)$; compute $\bm{s}=H_1\bm{e}$
  \State Set $\mu_x=\log\big(\tfrac{1-p_x}{p_x}\big)$ and initialize beliefs/messages
  \State Initialize $\hat{\bm{e}}$ and residual mismatch $\bm{\delta}=\bm{s}\oplus(H_1\hat{\bm{e}})$
  \For{BP-iteration $=1$ to $I_{\max}$}
    \State remaining $\leftarrow \{i| A_i>0\}$
    \For{step $=1$ to $|\text{remaining}|$}
      \If{$\bm{\delta}=\bm{0}$} \State \textbf{break} \EndIf
      \State Compute $\sigma_i$ for all $i\in$ remaining
      \State Choose $a\in$ remaining by $\varepsilon$-greedy on $Q(\sigma_i,i)$
      \State $w_{\text{before}}\leftarrow \|\bm{\delta}\|_1$
      \State Apply one SVNS update at $a$
      \State $w_{\text{after}}\leftarrow \|\bm{\delta}\|_1$
      \State $r_t\leftarrow (w_{\text{before}}-w_{\text{after}})/A_a$
      \State Remove $a$ from remaining
      \If{remaining is empty}
        \State best\_future $\leftarrow 0$
      \Else
        \State Recompute $\sigma_{i'}$ for remaining
        \State best\_future $\leftarrow \max_{i'\in\text{remaining}} Q(\sigma_{i'},i')$
      \EndIf
      \State $T \leftarrow r_t + \gamma\cdot \text{best\_future}$
      \State $Q(\sigma_a,a)\leftarrow Q(\sigma_a,a) + \alpha\big(T - Q(\sigma_a,a)\big)$
    \EndFor
  \EndFor
\EndFor
\end{algorithmic}
\end{algorithm}

\section{Fast Inference and Fast Monte-Carlo Evaluation}\label{sec:fast}

In this section, we define several implementation details to speed up the inference of RL-SVNS.\footnote{The methods described in this section present new and extend conventional implementation techniques for LDPC decoders to our setting and do not change the output of the RL-SVNS decoder, rather they address the implementation challenges of the RL trained schedule. As such the paper can be read without this section.} The inference-time RL-SVNS decoder is defined by the following features:
\begin{enumerate}
\item the same SVNS BP message-update equations (Subsection~\ref{subsec:SVNS}), and
\item the greedy RL policy 
\[
a=\arg\max_{i\in\text{remaining}}Q(\sigma_i,i).
\]
\end{enumerate}
In a baseline greedy RL-SVNS implementation, we identify the following computational bottlenecks:
\begin{itemize}
\item Repeated VN selection from the remaining set;
\item Repeated recomputation of local states $\sigma_i$ from neighboring residual bits $\delta_j$;
\item Edge-wise $\tanh$-product computations in message updates;
\item Frequent syndrome computations in Monte Carlo simulation (\ie, forming $\bm{s}=H_1\bm{e}$ and checking $H_1(\bm{e}\oplus\hat{\bm{e}})=0$).
\end{itemize}
We now describe how to reduce these bottlenecks by leveraging standard sparse-graph representations and incrementally updating key decoding quantities.

\subsection{Edge Indexing and Adjacency Arrays}\label{subsec:csr}

We represent the Tanner graph of the binary parity-check matrix
$H_1\in\{0,1\}^{m_1\times n}$
by an \emph{edge-indexed} description of its nonzero entries.
Let
\[
\mathcal{K}\triangleq\{(j,i)\in[m_1
]\times[n]\;|\; [H_1]_{j,i}=1\},\quad K\triangleq|\mathcal{K}|,
\]
where we use the notation $[k]\triangleq\{1,2,\ldots,k\}, k \in \mathbb{Z}^+,$ and $[A]_{j,i}$ to refer to the entry of matrix $A$ in row $j$ and column $i$. Fix an arbitrary but deterministic bijection
\begin{equation}\label{pi}
\pi:\{1,\ldots,K\}\to\mathcal{K},\quad
\pi(k)=(j_k,i_k), ~~ k=1,\ldots,K,
\end{equation}
so that each edge index $k$ corresponds to exactly one nonzero entry of $H_1$ (connecting CN $j_k$ to VN $i_k$).

\medskip
During decoding we repeatedly need to iterate over
(i) all edges incident to a given VN $i$, and
(ii) all edges incident to a given CN $j$.
If we store only the unordered list $\{(j_k,i_k)\}_{k=1}^K$, then identifying all incident edges of a node
can require scanning all $K$ edges.
Instead, we preprocess once and build \emph{adjacency arrays}:
two edge-index arrays, together with pointer (offset) arrays that locate each node's contiguous block of incident edges.

\medskip
\noindent\textbf{CN adjacency.}
For each CN $j\in[m_1]$, define its incident-edge set
\[
\mathcal{K}_j \;\triangleq\;\{k\in[K]\;|\; j_k=j\},\qquad
d^{(c)}_j \triangleq |\mathcal{K}_j|.
\]
We construct an array $\mathbf{k}^{(c)}\in[K]^K$ that stores (integer representations of) all edges incident to each CN contiguously,
and a pointer array $\mathbf{p}^{(c)}\in\{1,\ldots,K\}^{m}$ (defined below) such that
\begin{equation}\label{eq:csr_cn_math}
\mathcal{K}_j\ \leftrightarrow\
\mathbf{k}^{(c)}\!\big[\mathbf{p}^{(c)}_j:\mathbf{p}^{(c)}_j+d_j^{(c)}-1\big],
\end{equation}
where $[a:b]$ denotes an integer index range and $\mathbf{k}[a:b]$ denotes the sub-vector of $\mathbf{k}$ corresponding to that index range.
Thus, iterating over all neighbors of check $j$ becomes a tight loop over a contiguous slice of $\mathbf{k}^{(c)}$.

\medskip
\noindent\textbf{VN adjacency.}
Similarly, for each VN $i\in[n]$ define
\[
\mathcal{K}_i \;\triangleq\;\{k\in[K]\;|\; i_k=i\},\qquad
d^{(v)}_i \triangleq |\mathcal{K}_i|.
\]
We construct an array $\mathbf{k}^{(v)}\in[K]^K$ and pointers $\mathbf{p}^{(v)}\in\{1,\ldots,K\}^{n}$ such that
\begin{equation}\label{eq:csc_vn_math}
\mathcal{K}_i\ \leftrightarrow\
\mathbf{k}^{(v)}\!\big[\mathbf{p}^{(v)}_i:\mathbf{p}^{(v)}_i+d_i^{(v)}-1\big].
\end{equation}

\medskip
\noindent\textbf{Pointer construction.}
Let $d^{(c)}_j$ and $d^{(v)}_i$ be the CN and VN degrees defined above.
We build the pointers by prefix sums
\[
\mathbf{p}^{(c)}_1=1,\qquad
\mathbf{p}^{(c)}_{j+1}=\mathbf{p}^{(c)}_j+d^{(c)}_j,\quad j=1,\ldots,m-1,
\]
and analogously
\[
\mathbf{p}^{(v)}_1=1,\qquad
\mathbf{p}^{(v)}_{i+1}=\mathbf{p}^{(v)}_i+d^{(v)}_i,\quad i=1,\ldots,n-1.
\]
The arrays $\mathbf{k}^{(c)}$ and $\mathbf{k}^{(v)}$ are then filled by placing each edge index $k$ into the
next available slot of its corresponding block as described above.
After this preprocessing, the neighborhood iteration costs $O(d^{(v)}_i)$ (or $O(d^{(c)}_j)$) and accesses only contiguous memory.

Also, for each VN $i$, we fix a deterministic ordering of its incident edges
(nominally the increasing order of the check indices $j_k$).
Let $A_i \triangleq d^{(v)}_i$ and write the ordered incident edges as
\[
k_{i,0}, k_{i,1},\ldots,k_{i,A_i-1}\in\mathcal{K}_i .
\]
We assign a power-of-two weight to each edge according to its local position as
\begin{equation}\label{eq:edge_weight_def_math}
\beta(k_{i,t}) \triangleq 2^{t},\qquad t=0,1,\ldots,A_i-1 .
\end{equation}
These weights are \emph{local to VN $i$}: they map the $t$-th neighboring check of $i$ as the $t$-th bit of an
integer bitmask. This will be used in Section~\ref{subsec:incremental_state} to maintain local states via $O(1)$ XOR flips.

\begin{example}
Suppose a VN $i$ has degree $A_i=3$ and its neighboring CNs (in the chosen deterministic order) are
\[
\mathcal{N}(v_i)=\{j_1,j_2,j_3\},
\]
with ordered incident edges $k_{i,0}=(j_1,i)$, $k_{i,1}=(j_2,i)$, $k_{i,2}=(j_3,i)$.
Then \eqref{eq:edge_weight_def_math} assigns
\[
\beta(k_{i,0})=1,\qquad \beta(k_{i,1})=2,\qquad \beta(k_{i,2})=4.
\]
Later, if check $j_2$ flips its residual bit, then the state bit of VN $i$ corresponding to edge $k_{i,1}$
flips by XOR with $\beta(k_{i,1})=2$ (see \eqref{eq:xor_flip_math} in Section~\ref{subsec:incremental_state}).
\end{example}

After preprocessing, neighborhood traversals are simple range iterations over the slices as
\[
\mathcal{K}_i\ \leftrightarrow\ \mathbf{k}^{(v)}\!\big[\mathbf{p}^{(v)}_i:\mathbf{p}^{(v)}_{i+1}-1\big],\qquad
\mathcal{K}_j\ \leftrightarrow\ \mathbf{k}^{(c)}\!\big[\mathbf{p}^{(c)}_j:\mathbf{p}^{(c)}_{j+1}-1\big].
\]
This adjacency representation is the backbone of our fast inference implementation described later in Section~\ref{subsec:fast_inference}.

\subsection{Incremental Maintenance of Local States $\sigma_i$}\label{subsec:incremental_state}

Recall the residual mismatch bits
\[
\delta_j \in \{0,1\},\qquad j\in[m_1],
\]
defined from the current residual syndrome vector $\bm{\delta}$ (where  $\delta_j=1$ iff check $j$ is unsatisfied).
For each VN $i\in[n]$, we store its local RL state as an integer bitmask. 
Let $\beta:\mathcal{K}_i\to\{1,2,4,\dots\}$ be the per-edge power-of-two weight assignment as described in 
\eqref{eq:edge_weight_def_math} (local to each VN via the fixed order $k_{i,0},\ldots,k_{i,A_i-1}$).
We then define
\begin{equation}\label{eq:si_bitmask_math}
\sigma_i \;\triangleq\;\sum_{k\in \mathcal{K}_i} \delta_{j_k}\,\beta(k)\ \in\ \{0,1,\ldots,2^{A_i}-1\},
\end{equation}
so that the $t$-th bit of $\sigma_i$ (in its bit-wise binary notation) equals $1$ iff the $t$-th neighboring check of VN $i$ is currently unsatisfied.

Consider any CN $j\in[m_1]$. If its residual bit flips,
\[
\delta_j \leftarrow 1+\delta_j ~~ \text{in}~ \mathbb{F}_2,
\]
then for every incident edge $k\in\mathcal{K}$ with $j_k=j$ (so $\pi(k)=(j,i_k)$),
exactly one bit of the neighboring VN state must flip, namely the bit corresponding to edge $k$.
This update is implemented by a single bitwise XOR according to
\begin{equation}\label{eq:xor_flip_math}
\sigma_{i_k}\ \leftarrow\ \sigma_{i_k}\ \oplus_{2}\ \beta(k),
\end{equation}
where $\oplus_{2}$ denotes bitwise XOR on the integer masks.
Indeed, in \eqref{eq:si_bitmask_math} the term $\delta_{j_k}\beta(k)$ flips between $0$ and $\beta(k)$, which is
precisely the effect of XOR with $\beta(k)$ on the corresponding bit.

Now suppose an SVNS update at some VN $i$ flips the current hard decision $\hat e_i$.
By Lemma~\ref{lem:delta_flip}, only checks adjacent to $i$ can change their residual bits, \ie, $\delta_j$  may flip only for $j\in \mathcal{N}(v_i)$ and unchanged for $j\notin \mathcal{N}(v_i),$
where the check-neighborhood of VN $i$ is
\[
\mathcal{N}(v_i)\triangleq \{j\in[m_1]\;|\; \exists\, k\in\mathcal{K}\ \text{s.t.}\ (j_k,i_k)=(j,i)\}.
\]
Consequently, only VNs adjacent to those flipped checks can have their states changed. Define the
affected VN set (second-order neighborhood)
\[
\mathcal{V}_{\mathrm{aff}}(i)\ \triangleq\ \bigcup_{j\in\mathcal{N}(v_i)} \mathcal{N}(c_j),
\]
where
\[
\mathcal{N}(c_j)\triangleq \{u\in[n]\;|\; \exists\, k\in\mathcal{K}\ \text{s.t.}\ (j_k,i_k)=(j,u)\}.
\]
For sparse QLDPC graphs, $|\mathcal{V}_{\mathrm{aff}}(i)|\ll n$, so we update only $\{\sigma_u\;|\; u\in\mathcal{V}_{\mathrm{aff}}(i)\}$
using the XOR rule \eqref{eq:xor_flip_math}, rather than recomputing all states from scratch.

\begin{example}
Assume $A_{\max}=3$ and consider the neighborhood
\[
\begin{aligned}
\mathcal{N}(v_i)   &= \{j_4,j_7\},\\
\mathcal{N}(c_{j_4}) &= \{i,u_1,u_2\},\qquad
\mathcal{N}(c_{j_7}) = \{i,u_2,u_3\}.
\end{aligned}
\]
Fix the deterministic neighbor order at each VN and the corresponding weights $\beta(\cdot)$.
For this example, assume the local orders (hence bit positions) are
\[
\begin{aligned}
&\text{VN } i: &&(j_4,j_7)\ \Rightarrow\ \beta(k_{j_4,i})=1,\ \beta(k_{j_7,i})=2,\\
&\text{VN } u_1: &&(j_4,j_9)\ \Rightarrow\ \beta(k_{j_4,u_1})=1,\ \beta(k_{j_9,u_1})=2,\\
&\text{VN } u_2: &&(j_4,j_7,j_{10})\ \Rightarrow\ \beta(k_{j_4,u_2})=1,\ \beta(k_{j_7,u_2})=2,\\
&            &&\hspace{3.2em}\beta(k_{j_{10},u_2})=4,\\
&\text{VN } u_3: &&(j_7,j_{11})\ \Rightarrow\ \beta(k_{j_7,u_3})=1,\ \beta(k_{j_{11},u_3})=2,
\end{aligned}
\]
where we refer to edge $k_{j_a,i_b}$ meaning $\pi(k_{j_a,i_b})=(j_a,i_b)$ from \eqref{pi}. Assume the current residual bits are
\[
\delta_{j_4}=1,~~ \delta_{j_7}=0,~~ \delta_{j_9}=1,~~ \delta_{j_{10}}=0,~~ \delta_{j_{11}}=1.
\]
Then the VN states from \eqref{eq:si_bitmask_math} are
\[
\begin{aligned}
\sigma_i   &= 1\cdot \delta_{j_4} + 2\cdot \delta_{j_7} = 1,\\
\sigma_{u_1} &= 1\cdot \delta_{j_4} + 2\cdot \delta_{j_9} = 3,\\
\sigma_{u_2} &= 1\cdot \delta_{j_4} + 2\cdot \delta_{j_7} + 4\cdot \delta_{j_{10}} = 1,\\
\sigma_{u_3} &= 1\cdot \delta_{j_7} + 2\cdot \delta_{j_{11}} = 2.
\end{aligned}
\]

\smallskip
Now suppose the decoder updates VN $i$ and the hard decision flips. Then only the adjacent checks flip (parity change) according to
\[
\delta_{j_4}:1\to 0,\qquad \delta_{j_7}:0\to 1,
\]
while all other $\delta_j$ remain unchanged. Instead of recomputing all states, we apply the XOR updates
\eqref{eq:xor_flip_math} for every edge incident to the flipped checks.\footnote{We note that the sum-product algorithm may not ``flip'' a check with soft messages; however, we remind the reader that to determine the state of a VN we only use hard decisions.} {Updates due to $j_4$ flipping are given by}
\[
\sigma_i \leftarrow \sigma_i \oplus_{2} 1,\qquad
\sigma_{u_1} \leftarrow \sigma_{u_1} \oplus_{2} 1,\qquad
\sigma_{u_2} \leftarrow \sigma_{u_2} \oplus_{2} 1,
\]
{and updates due to $j_7$ flipping are given by}
\[
\sigma_i \leftarrow \sigma_i \oplus_{2} 2,\qquad
\sigma_{u_2} \leftarrow \sigma_{u_2} \oplus_{2} 2,\qquad
\sigma_{u_3} \leftarrow \sigma_{u_3} \oplus_{2} 1,
\]
where the last XOR uses $\beta(k_{j_7,u_3})=1$ because $j_7$ is the first neighbor of $u_3$ in its fixed order. Carrying out the XORs gives
\[
\begin{aligned}
\sigma_i:   &\ 1 \oplus_{2} 1 \oplus_{2} 2 = 2,\\
\sigma_{u_1}:&\ 3 \oplus_{2} 1 = 2,\\
\sigma_{u_2}:&\ 1 \oplus_{2} 1 \oplus_{2} 2 = 2,\\
\sigma_{u_3}:&\ 2 \oplus_{2} 1 = 3,
\end{aligned}
\]
which matches direct recomputation using the new residual bits $(\delta_{j_4},\delta_{j_7})=(0,1)$.

We note that only VNs in $\mathcal{N}(c_{j_4})\cup\mathcal{N}(c_{j_7})=\{i,u_1,u_2,u_3\}$ were modified.
All other VNs keep the same state because they are not adjacent to the flipped checks.
This is exactly the second-order neighborhood effect: a flip at VN $i$ changes checks in $\mathcal{N}(v_i)$ (one order),
which can only affect VNs in $\bigcup_{j\in\mathcal{N}(v_i)}\mathcal{N}(c_j)$ (second order) and the VN $i$ itself.
\end{example}

\subsection{SVNS Check-products (Cached $\tanh$-products)}\label{subsec:cached_products}

In the SVNS update, each CN-to-VN message on an edge $k=(j,i)$ involves a product over
\emph{all other} VN-to-CN terms incident to the same CN $j$.
If this product is recomputed from scratch for every incident edge, then the work at check $j$
scales as $O\!\big((d^{(c)}_j)^2\big)$ per sweep, where $d^{(c)}_j$ is the degree of check $j$.
This redundancy can be eliminated by caching one product per check.

Using the edge indexing described in Section~\ref{subsec:csr}, let $\pi(k)=(j_k,i_k)$ for $k\in[K]$,
and with the incident-edge sets $\mathcal{K}_j$ and $\mathcal{K}_i$.
We store one VN-to-CN LLR message per edge as
\[
m^{(v\to c)}_k \in \mathbb{R},\qquad k\in[K].
\]
As usual, the LLR-domain CN update is expressed via the $\tanh(\cdot/2)$ transform. Define
\begin{equation}\label{eq:x_edge_def_math}
x_k \triangleq \tanh\!\Big(\tfrac{m^{(v\to c)}_k}{2}\Big),\qquad k\in[K].
\end{equation}
For each CN $j\in[m_1]$, we cache the product of all incident $x$-values according to
\begin{equation}\label{eq:Pc_def_math}
P_j \triangleq \prod_{k\in\mathcal{K}_j} x_k .
\end{equation}
Then, for an edge $k\in\mathcal{K}_j$, the excluded product
\begin{equation}\label{eq:prod_excl_def_math}
\prod_{k'\in\mathcal{K}_j\setminus\{k\}} x_{k'}
\end{equation}
is ideally obtained by the ratio
\begin{equation}\label{eq:prod_excl_div_math}
\operatorname{prod\_excl}(j,k)\triangleq \frac{P_j}{x_k}.
\end{equation}

If $|x_k|$ is extremely small, the division in \eqref{eq:prod_excl_div_math} can be unstable.
In our numerical results we fix a small threshold $\varepsilon_{th}>0$ and decide:
\begin{itemize}
\item if $|x_k|>\varepsilon_{th}$, use \eqref{eq:prod_excl_div_math};
\item otherwise, recompute \eqref{eq:prod_excl_def_math} directly by a single loop over
$\mathcal{K}_j\setminus\{k\}$.
\end{itemize}
This yields the \emph{same mathematical CN update} as the baseline implementation,
while avoiding redundant products in typical regimes where $|x_k|$ rarely falls below $\varepsilon_{th}$.

We now define some modification for efficiency to to steps 2-4 from the conventional BP described in Section~\ref{subsec:SVNS}.Let $s_j\in\{0,1\}$ denote the measured syndrome bit at check $j$.
The CN-to-VN LLR message along edge $k\in\mathcal{K}_j$ (from CN $j$ to VN $i_k$) is
\begin{equation}\label{eq:c2v_cached_math}
m^{(c\to v)}_{j\to i_k}
=
2\,\operatorname{atanh}\!\Big(
(-1)^{s_j}\,\operatorname{prod\_excl}(j,k)
\Big),
\end{equation}
followed by clipping to $[-L_{\max},L_{\max}]$ to prevent overflow and to keep $|\tanh(\cdot/2)|$ away from $1$. When SVNS selects a VN $i$, we form its updated a-posteriori LLR as
\begin{equation}\label{eq:Lnew_cached_math}
L_i^{\mathrm{new}}
=
L_i \;+\; \sum_{k\in\mathcal{K}_i} m^{(c\to v)}_{j_k\to i},
\end{equation}
where $L_i$ is the channel LLR.
The hard decision is updated as $\hat e_i \leftarrow \mathds{1}[L_i^{\mathrm{new}}<0]$. We then update the outgoing variable-to-check messages on incident edges using the standard extrinsic form
\begin{equation}\label{eq:v2c_update_cached_math}
m^{(v\to c)}_k
\leftarrow
\mathrm{clip}\!\Big(
L_i^{\mathrm{new}} - m^{(c\to v)}_{j_k\to i},\;[-L_{\max},L_{\max}]
\Big),\qquad k\in\mathcal{K}_i,
\end{equation}
and refresh $x_k$ via \eqref{eq:x_edge_def_math}.

Throughout, we use an incremental maintenance of the CN cache to update $P_j$, \ie, only the cache of check $j$ must be updated whenever $x_k$ changes on an edge $k$ with $\pi(k)=(j,i)$.
Let $x_k^{\mathrm{old}}$ and $x_k^{\mathrm{new}}$ denote the values before and after \eqref{eq:v2c_update_cached_math}.
If $|x_k^{\mathrm{old}}|>\varepsilon_{th}$, we update in $O(1)$ complexity using
\begin{equation}\label{eq:Pc_update_ratio_math}
P_j \leftarrow P_j \cdot \frac{x_k^{\mathrm{new}}}{x_k^{\mathrm{old}}}.
\end{equation}
If $|x_k^{\mathrm{old}}|\le\varepsilon_{th}$, we recompute $P_j$ directly by one loop over $\mathcal{E}_j$
to maintain robustness. A single SVNS update at VN $i$ modifies only:
(i) the edge messages $\{m^{(v\to c)}_k\;|\; k\in\mathcal{K}_i\}$ (hence $\{x_k\;|\; k\in\mathcal{K}_i\}$), and
(ii) the cached products $\{P_j\;|\; j\in\mathcal{N}(v_i)\}$.
All other edges and check caches remain unchanged.
Thus, caching preserves the exact BP update rule while removing redundant per-edge product computations at CNs.

\subsection{Fast Update After a Hard-decision Flip (Residual, States, and Priorities)}\label{subsec:flip_update}

Within one SVNS iteration, the greedy RL scheduler selects the next VN among those not yet processed
(schedule-without-replacement) by comparing the scores $Q(s_u,u)$.
Hence, when an SVNS update at a VN $i$ flips its hard decision $\hat e_i$,
we must update only the quantities that can change the scores of the remaining nodes:
the residual bits $\{\delta_j\}$ (equivalently the residual syndrome), the mismatch weight $w$, 
the local states $\{\sigma_u\}$, and the priorities of the remaining nodes used by the greedy selection mechanism.

What follows is a local implementation of the steps defined in the previopus subsections. If $\hat e_i$ flips, Lemma~\ref{lem:delta_flip} implies that only CNs adjacent to $i$ can flip residual status: $\delta_j$ may flip only for $j\in\mathcal{N}(v_i)$ and unchanged for $j\notin\mathcal{N}(v_i)$. 
Consequently, only VNs adjacent to those flipped checks can see their local states change. 
For sparse Tanner graphs,
$|\mathcal{V}_{\mathrm{aff}}(v)|\ll n$. 
Each flipped residual bit changes $w$ by exactly $\pm 1$ as
\[
\delta_j \leftarrow 1\oplus \delta_j
\quad\Longrightarrow\quad
w \leftarrow w + (2\delta_j-1),
\]
where $\delta_j$ on the right-hand side denotes the \emph{new} value after flipping.
Thus, updating $w$ after a flip at $i$ costs $O(|\mathcal{N}(i)|)$. Local states $\sigma_u$ are maintained by the XOR rule of Section~\ref{subsec:incremental_state}:
whenever a residual bit $\delta_j$ flips, each incident edge $k\in\mathcal{K}_j$ triggers an update of $\sigma_{i_k}$ following \eqref{eq:xor_flip_math}. 
Therefore, when $\hat e_i$ flips, it suffices to apply these XOR flips for all edges incident
to each check in $\mathcal{N}(v_i)$, which touches exactly the nodes in $\mathcal{V}_{\mathrm{aff}}(i)$.

Let $\mathcal{R}$ denote the set of VNs that are still \emph{remaining} in the current iteration.
Only nodes $u\in\mathcal{R}\cap \mathcal{V}_{\mathrm{aff}}(v)$ can have changed state, hence only these nodes can have changed scores.
Accordingly, we recompute and update the priority key of the heap (explained further in \ref{subsec:heap}) only for this local subset as
\begin{equation}\label{eq:key_update_local}
\mathrm{key}(u) \leftarrow Q(\sigma_u,u),\qquad u\in\mathcal{R}\cap \mathcal{V}_{\mathrm{aff}}(i),
\end{equation}
followed by a local data-structure repair.

Since a flip at $i$ has strictly local consequences: it can only change residual bits on $\mathcal{N}(v_i)$ and only change
states/scores of VNs in $\mathcal{V}_{\mathrm{aff}}(v_i)$.
Hence, greedy scheduling requires priority updates only inside this second-order neighborhood.

\subsection{Greedy Selection with a Max-heap (Priority Queue)}\label{subsec:heap}

At each scheduling step we must select an index achieving
\begin{equation}\label{eq:greedy_argmax}
u^\star \in \arg\max_{u\in\mathcal{R}} Q(s_u,u),
\end{equation}
where $\mathcal{R}$ is the current remaining set.
A naive implementation evaluates \eqref{eq:greedy_argmax} by scanning all $u\in\mathcal{R}$, costing $O(|\mathcal{R}|)$ time per step.
We maintain a max-heap over the indices in $\mathcal{R}$ with key function
\begin{equation}\label{eq:heap_key_def}
\mathrm{key}(u)\triangleq Q(s_u,u).
\end{equation}
A max-heap supports:
(i) \emph{extract-max} (remove and return an index with maximal key) in $O(\log|\mathcal{R}|)$ time, and
(ii) \emph{change-key} (update the key of a specific item and restore heap order) in $O(\log|\mathcal{R}|)$ time.

At the beginning of each SVNS iteration we build a heap over the active variables.
Each scheduling step applies extract-max once; the extracted node is removed from $\mathcal{R}$ and is not reinserted
within the same iteration.
When a flip-update changes some states $\sigma_u$, we update keys only for the affected remaining nodes
$u\in\mathcal{R}\cap\mathcal{V}_{\mathrm{aff}}(i)$ as in \eqref{eq:key_update_local}, and then repair heap order locally
via a change-key operation.
Heap-based greedy selection results for each argmax in \eqref{eq:greedy_argmax} to cost $O(\log|\mathcal{R}|)$ time,
and state-induced key updates cost $O(\log|\mathcal{R}|)$ time per affected remaining node only, matching the second-order locality
of Section~\ref{subsec:flip_update}.

\subsection{Main Inference Loop}\label{subsec:fast_inference}

We now summarize how the implementation combines:
(i) edge indexing and adjacency arrays (Section~\ref{subsec:csr}),
(ii) cached CN products for the SVNS CN update (Section~\ref{subsec:cached_products}),
(iii) incremental maintenance of local states (Section~\ref{subsec:incremental_state}),
(iv) flip-triggered local updates of residuals, states, and priorities (Section~\ref{subsec:flip_update}),
and (v) greedy selection via a max-heap (Section~\ref{subsec:heap}).

\paragraph{Initialization}
Given channel LLRs $\{L_i\}_{i=1}^n$ and measured syndrome bits $\{s_j\}_{j=1}^{m_1}$:
\begin{itemize}
\item Initialize edge messages $\{m^{(i\to j)}_k\}_{k\in[K]}$ from the channel LLRs, and form $\{x_k\}$ via \eqref{eq:x_edge_def_math};.
\item Build cached products $\{P_j\}$ via \eqref{eq:Pc_def_math};
\item Initialize hard decisions $\hat e_i=\mathds{1}[L_i<0]$ and residual bits $\delta_j$ by comparing the measured
syndrome $s_j$ to the predicted syndrome induced by $\hat e$ (XOR accumulation over $\mathcal{N}(j)$).
Set $w=\sum_j \delta_j$;
\item Initialize local states $\{\sigma_i\}$ from \eqref{eq:si_bitmask_math}.
\end{itemize}

\paragraph{One BP iteration}
Fix an iteration cap $I_{\max}$. For each iteration:
\begin{enumerate}
\item If $w=0$, declare success;
\item Initialize the remaining set $\mathcal{R}$ to the active VNs and build a max-heap with keys $\mathrm{key}(u)=Q(s_u,u)$;
\item While $\mathcal{R}\neq\emptyset$:
\begin{itemize}
\item Extract $u^\star$ with largest key (Section~\ref{subsec:heap}) and remove it from $\mathcal{R}$;
\item Perform one SVNS update at $u^\star$ using cached products (Section~\ref{subsec:cached_products}),
updating only $\{m^{(i\to j)}_k\;|\; k\in\mathcal{K}_{u^\star}\}$ and $\{P_j\;|\; j\in\mathcal{N}(u^\star)\}$;
\item If the hard decision at $u^\star$ flips, apply the local residual/state/priority updates of Section~\ref{subsec:flip_update}.
\end{itemize}
\end{enumerate}
If the cap is reached with $w>0$, declare failure.

\section{RL Scheduling for Quaternary SVNS Decoding}
\label{sec:rl_svns_depolarizing}

We now extend the RL scheduling framework from an $X$-only channel to the depolarizing channel, while keeping the underlying message-update rule fixed (SVNS) and learning only \emph{which} VN to update next.
In particular, we keep the two-stream SVNS CN/VN updates, but we (i) compute VN-to-CN messages via a quaternary-coupled extrinsic belief and (ii) make the hard decision depend on the current (freshly updated) LLR pair at the visited node.
{Our message updates follow the scalar-message reparameterization of conventional quaternary BP over $\{I,X,Y,Z\}$ \cite{poulin2008iterative}, as derived in \cite{kuo2020refined, lai2021log}.}

\subsection{CSS Framework and Depolarizing Prior}
\label{subsec:css_depolarizing_prior}

Consider a CSS code specified by binary matrices $(H_X,H_Z)$ and a Pauli error pattern
$\mathbf{Q}=(Q_1,\dots,Q_n)\in\{I,X,Y,Z\}^n$.
We decompose $\mathbf{Q}$ into two binary error vectors
$\mathbf{e}^X,\mathbf{e}^Z\in\mathbb{F}_2^n$ defined componentwise by
\[
e^X_i \triangleq \mathbf{1}[Q_i\in\{Z,Y\}],
\qquad
e^Z_i \triangleq \mathbf{1}[Q_i\in\{X,Y\}],
\]
so that the measured syndromes decompose as
\[
\mathbf{s}^X = H_X \mathbf{e}^X,\qquad
\mathbf{s}^Z = H_Z \mathbf{e}^Z,
\]
where decoding aims to find a Pauli estimate consistent with $(\mathbf{s}^X,\mathbf{s}^Z)$.

Under depolarizing noise with physical error rate $p$, each qubit satisfies
\[
\begin{aligned}
&\Pr(Q_i=I)=1-p,\\
&\Pr(Q_i=X)=\Pr(Q_i=Y)=\Pr(Q_i=Z)=\frac{p}{3}.
\end{aligned}
\]
We define the Tanner-graph neighborhoods
\[
\begin{aligned}
\mathcal{N}_X(v_i) &\triangleq \{j\;|\; H_X(j,i)=1\},\\
\mathcal{N}_Z(v_i) &\triangleq \{j\;|\; H_Z(j,i)=1\}.
\end{aligned}
\]
and similarly $\mathcal{N}_X(c_j)$ (resp. $\mathcal{N}_Z(c_j)$) for the variable neighbors of check $c_j$ in
the $H_X$ (resp. $H_Z$) graph.
Let $A_i^X \triangleq |\mathcal{N}_X(v_i)|$ and $A_i^Z \triangleq |\mathcal{N}_Z(v_i)|$.

\subsection{SVNS Updates}
\label{subsec:two_stream_svns}

Rather than passing full 4-ary edge messages, we maintain two scalar LLR streams:
one on the $H_X$ graph and one on the $H_Z$ graph.
For each edge $(j,i)$ in the $H_X$ Tanner graph, store LLR messages
$m^{X}_{i\to j}$ and $m^{X}_{j\to i}$; for each edge $(j,i)$ in the $H_Z$ Tanner graph, store
$m^{Z}_{i\to j}$ and $m^{Z}_{j\to i}$. For all edges incident to VN $v_i$, we initialize the VN-to-CN LLRs as
\begin{align}
m^{X}_{i\to j} &\triangleq \log\frac{\Pr(Q_i\in\{I,X\})}{\Pr(Q_i\in\{Y,Z\})}
= \log\frac{1-\tfrac{2p}{3}}{\tfrac{2p}{3}}
\;\triangleq\; \mu_{\rm dep}, \label{eq:init_mx}\\
m^{Z}_{i\to j} &\triangleq \log\frac{\Pr(Q_i\in\{I,Z\})}{\Pr(Q_i\in\{X,Y\})}
= \log\frac{1-\tfrac{2p}{3}}{\tfrac{2p}{3}}
\;=\; \mu_{\rm dep}. \label{eq:init_mz}
\end{align}

A single SVNS step at VN $i$ consists of the following BP-style refreshes on \emph{both} Tanner graphs:

\smallskip
\noindent\textbf{Step 1)}
For each $j\in\mathcal{N}_X(v_i)$, update the $H_X$-side CN-to-VN message using 
\begin{equation}
m^{X}_{j\to i} \;=\; 2\,\tanh^{-1}\!\left(
(-1)^{s^X_j}\!\!\!\prod_{i'\in\mathcal{N}_X(c_j)\setminus\{i\}}
\tanh\!\left(\frac{m^{X}_{i'\to j}}{2}\right)\right).
\label{eq:cn_update_x}
\end{equation}
Similarly, for each $j\in\mathcal{N}_Z(v_i)$,
\begin{equation}
m^{Z}_{j\to i} \;=\; 2\,\tanh^{-1}\!\left(
(-1)^{s^Z_j}\!\!\!\prod_{i'\in\mathcal{N}_Z(c_j)\setminus\{i\}}
\tanh\!\left(\frac{m^{Z}_{i'\to j}}{2}\right)\right);
\label{eq:cn_update_z}
\end{equation}

\smallskip
\noindent\textbf{Step 2)}
Form the two a-posteriori LLRs
\begin{equation}
L_i^{X} \;=\; \mu_{\rm dep} + \sum_{j\in\mathcal{N}_X(v_i)} m^{X}_{j\to i},
~~
L_i^{Z} \;=\; \mu_{\rm dep} + \sum_{j\in\mathcal{N}_Z(v_i)} m^{Z}_{j\to i};
\label{eq:Li_streams}
\end{equation}

\smallskip
\noindent\textbf{Step 3)}
Update VN-to-CN messages using a \emph{quaternary-coupled} extrinsic belief. For each Pauli error $q\in\{I,X,Y,Z\}$, define the component bits
\begin{equation}
e^{X}(q)\triangleq \mathbf{1}[q\in\{Y,Z\}],
\qquad
e^{Z}(q)\triangleq \mathbf{1}[q\in\{X,Y\}].
\label{eq:component_bits}
\end{equation}
Under depolarizing noise, the induced \emph{marginal} component priors are
\begin{equation}
\pi^{X}_1=\Pr(e^{X}=1)=\Pr(Q\in\{Y,Z\})=\frac{2p}{3},
~~
\pi^{X}_0=1-\frac{2p}{3},
\label{eq:piX_marginal}
\end{equation}
and similarly $\pi^{Z}_1=\frac{2p}{3}$ and $\pi^{Z}_0=1-\frac{2p}{3}$. To avoid double counting the marginals already injected by $\mu_{\rm dep}$ while restoring the joint depolarizing prior (the $Y$-correlation), define the correction factor
\begin{equation}
\kappa_q(p)\triangleq
\frac{\pi_q}{\pi^{X}_{e^{X}(q)}\,\pi^{Z}_{e^{Z}(q)}},
~~
\pi_I=1-p,\;\;\pi_X=\pi_Y=\pi_Z=\frac{p}{3}.
\label{eq:kappa_def}
\end{equation}

Given any LLR pair $(L^{Z},L^{X})$, define the quaternary likelihood map
\begin{equation}
\Phi_i\!\big(q;L^{Z},L^{X}\big)\triangleq
\exp\!\Big(
\tfrac{1}{2}(1-2e^{Z}(q))\,L^{Z}
+\tfrac{1}{2}(1-2e^{X}(q))\,L^{X}
\Big).
\label{eq:Phi_def}
\end{equation}
Now form the extrinsic LLR pairs for each neighboring check $j$ by computing
\begin{equation}
L^{X\setminus j}_i \;\triangleq\; L_i^X - m^X_{j\to i},
\qquad
L^{Z\setminus j}_i \;\triangleq\; L_i^Z - m^Z_{j\to i}.
\label{eq:extrinsic_pairs}
\end{equation}
Using these pairs, form the unnormalized extrinsic quaternary scores
(for $q\in\{I,X,Y,Z\}$) according to
\begin{align}
\tilde B^{X}_{i\to j}(q) &\triangleq \kappa_q(p)\,
\Phi_i\!\big(q;\,L_i^{Z},\,L_i^{X\setminus j}\big),
\qquad j\in\mathcal{N}_X(v_i), \label{eq:Bex_X}\\
\tilde B^{Z}_{i\to j}(q) &\triangleq \kappa_q(p)\,
\Phi_i\!\big(q;\,L_i^{Z\setminus j},\,L_i^{X}\big),
\qquad j\in\mathcal{N}_Z(v_i). \label{eq:Bex_Z}
\end{align}

Finally, we update outgoing LLRs by grouping Pauli operators according to the component bit seen on each Tanner graph as
\begin{equation}
\begin{aligned}
m^{X}_{i\to j}
&\;=\; \log\frac{\tilde B^{X}_{i\to j}(I)+\tilde B^{X}_{i\to j}(X)}
{\tilde B^{X}_{i\to j}(Y)+\tilde B^{X}_{i\to j}(Z)},\\
m^{Z}_{i\to j}
&\;=\; \log\frac{\tilde B^{Z}_{i\to j}(I)+\tilde B^{Z}_{i\to j}(Z)}
{\tilde B^{Z}_{i\to j}(X)+\tilde B^{Z}_{i\to j}(Y)};
\end{aligned}
\label{eq:vn_extrinsic}
\end{equation}

\smallskip
\noindent\textbf{Step 4)}
Now, we combine the two streams into a quaternary belief and hard decision \emph{using the freshly updated pair} $(L_i^X,L_i^Z)$ at the visited VN. 
Using \eqref{eq:kappa_def}-\eqref{eq:Phi_def}, we define
\[
\tilde B_i(q)\triangleq \kappa_q(p)\,\Phi_i\!\big(q;L_i^{Z},L_i^{X}\big),
\qquad q\in\{I,X,Y,Z\}.
\]
Normalization gives
\[
B_i(q)=\frac{\tilde B_i(q)}{\sum_{q'\in\{I,X,Y,Z\}}\tilde B_i(q')}.
\]
At each visit of VN $i$, we update the hard decision using the \emph{current} belief (after Steps~1-3 at node $i$) according to
\[
\hat q_i=\arg\max_{q\in\{I,X,Y,Z\}} B_i(q).
\]
We then define
\[
\hat e^Z_i \triangleq \mathbf{1}[\hat q_i\in\{X,Y\}],
\qquad
\hat e^X_i \triangleq \mathbf{1}[\hat q_i\in\{Z,Y\}].
\]

\subsection{Residual Mismatches and Locality Under Pauli Errors}
\label{subsec:residual_locality_depolarizing}

Let $\hat{\mathbf{e}}^X,\hat{\mathbf{e}}^Z\in\mathbb{F}_2^n$ denote the current hard estimates induced by
$\hat{\mathbf{q}}=(\hat q_1,\dots,\hat q_n)$.
Let the residual mismatch vectors be
\begin{equation}
\boldsymbol{\delta}^{X} \;\triangleq\; \mathbf{s}^{X}\oplus (H_X \hat{\mathbf{e}}^X),
\qquad
\boldsymbol{\delta}^{Z} \;\triangleq\; \mathbf{s}^{Z}\oplus (H_Z \hat{\mathbf{e}}^Z),
\label{eq:deltaXZ}
\end{equation}
and the scalar mismatch weight
\begin{equation}
w \;\triangleq\; \|\boldsymbol{\delta}^{X}\|_1 + \|\boldsymbol{\delta}^{Z}\|_1.
\label{eq:w_depolarizing}
\end{equation}

When the hard decision at a \emph{single} VN $i$ changes, only a local subset of check residuals can flip, resulting in the following possibilities:
\begin{itemize}
\item If $\hat e^X_i$ flips (\ie, the decision changes between $\{I,X\}$ and $\{Y,Z\}$), then only
checks in $\mathcal{N}_X(v_i)$ can flip in $\boldsymbol{\delta}^{X}$;
\item If $\hat e^Z_i$ flips (\ie, the decision changes between $\{I,Z\}$ and $\{X,Y\}$), then only
checks in $\mathcal{N}_Z(v_i)$ can flip in $\boldsymbol{\delta}^{Z}$;
\item If both components flip (\eg, $I\leftrightarrow Y$ or $X\leftrightarrow Z$), both neighborhoods will change.
\end{itemize}
This locality enables efficient incremental maintenance of the RL state features defined next.

\subsection{RL State, Action, and Reward Under Depolarizing}
\label{subsec:rl_state_action_reward_depolarizing}

For each VN $v_i$, let there be local residual patterns
\[
\begin{aligned}
\mathbf{b}_i^{X} &\triangleq \big[\boldsymbol{\delta}^{X}\big|_{\mathcal{N}_X(v_i)},\,0,\ldots,0\big]\in\{0,1\}^{A_{\max}},\\
\mathbf{b}_i^{Z} &\triangleq \big[\boldsymbol{\delta}^{Z}\big|_{\mathcal{N}_Z(v_i)},\,0,\ldots,0\big]\in\{0,1\}^{A_{\max}},
\end{aligned}
\]
using any fixed ordering of neighbors, where we take a \emph{single} padding length
$A_{\max}\triangleq \max\{A_{\max}^{X},A_{\max}^{Z}\}$ (with $A_{\max}^{X}\triangleq \max_i A_i^X$ and
$A_{\max}^{Z}\triangleq \max_i A_i^Z$) and pad the shorter side with zeros.
By representing each pattern as a binary integer
\[
\sigma_i^{X} \triangleq \sum_{k=0}^{A_{\max}-1} b^{X}_{i,k}\,2^k,
\qquad
\sigma_i^{Z} \triangleq \sum_{k=0}^{A_{\max}-1} b^{Z}_{i,k}\,2^k,
\]
we combine them into a single state integer id
\begin{equation}
\sigma_i \;\triangleq\; \sigma_i^{X} + 2^{A_{\max}}\, \sigma_i^{Z},
\label{eq:state_joint}
\end{equation}
corresponding to the binary representation of $[\mathbf{b}_i^{Z},\mathbf{b}_i^{X}]$. Inside each BP iteration, the agent maintains a remaining set
$\mathcal{R}\subseteq\{1,\dots,n\}$ and chooses the next VN
$a\in\mathcal{R}$ to apply the SVNS update described in
Subsection~\ref{subsec:two_stream_svns}.
After selecting $a$, we remove it from $\mathcal{R}$.

Let $w_{\rm before}$ and $w_{\rm after}$ be the mismatch weights \eqref{eq:w_depolarizing}
immediately before and after applying the SVNS update at VN $a$
(including any consequent hard-decision change and residual updates).
We use the normalized reward
\begin{equation}
r_t \;=\; \frac{w_{\rm before}-w_{\rm after}}{A_a^{X}+A_a^{Z}},
\label{eq:reward_depolarizing}
\end{equation}
with a terminal bonus $+1$ if $w_{\rm after}=0$ (both syndromes satisfied). Using the same schedule lookahead as in the $X$-only setting,
we define
\[
\text{best\_future}\;=\;\max_{a'\in\mathcal{R}} Q(\sigma_{a'},a'),
\]
and update the tabular value function as
\[
Q(\sigma_a,a)\leftarrow (1-\alpha)Q(\sigma_a,a) + \alpha\Big(r_t+\gamma\,\text{best\_future}\Big).
\]
After learning, for inference we use the greedy policy
$a=\arg\max_{i\in\mathcal{R}}Q(\sigma_i,i)$.

Relative to the $X$-only RL-SVNS decoder, the only structural change under depolarizing noise is that
(i) each SVNS step updates \emph{two} LLR streams \eqref{eq:cn_update_x}-\eqref{eq:vn_extrinsic},
(ii) hard decisions are quaternary, yielding $(\hat{\mathbf{e}}^X,\hat{\mathbf{e}}^Z)$, and
(iii) the RL state and reward are driven by the \emph{pair} of residual mismatch vectors
\eqref{eq:deltaXZ} through \eqref{eq:w_depolarizing}-\eqref{eq:reward_depolarizing}.\footnote{Note that the state space is significantly larger than that of the X-only case; however, in our experiments, this was manageable for practical codes of interest.}
In addition, the outgoing VN$\to$CN LLRs are computed
from a quaternary-coupled extrinsic belief that depends on both streams (via \eqref{eq:Bex_X}-\eqref{eq:vn_extrinsic}).
All other RL components (action space, schedule-without-replacement, and Q-learning) remain unchanged.

\section{Numerical Results}\label{sec:numerical}

In our RL training (offline), we sample the channel crossover probability $p_x$ ($p$ for quaternary decoding) uniformly from the training grid
$P_{\text{train}}=\{0.03,0.04,0.05,0.06,0.07\}$, and run the SVNS decoder with learned scheduling for at most
$I_{\max} =100$ BP iterations per episode. We train for $E_{\max}=10^5$ episodes. The Q-table is initialized to zero and
updated using tabular Q-learning with learning rate $\alpha=0.1$ and discount factor $\gamma=0.9$.
Exploration follows an $\varepsilon$-greedy policy with $\varepsilon_0=0.6$, a minimum value
$\varepsilon_{\min}=0.05$, and a linear decay over episodes,
\[
\varepsilon_{\mathrm{iter}}
=\max\!\left(\varepsilon_{\min},\,\varepsilon_0\Big(1-\frac{\mathrm{iter}-1}{E_{\max}-1}\Big)\right),
\]
for $E_{\max}$ iterations.

\subsection{RL trained SVNS for Independent Pauli-$X$ Errors}

\begin{figure}[t]
  \centering
  \includegraphics[width=1\linewidth]{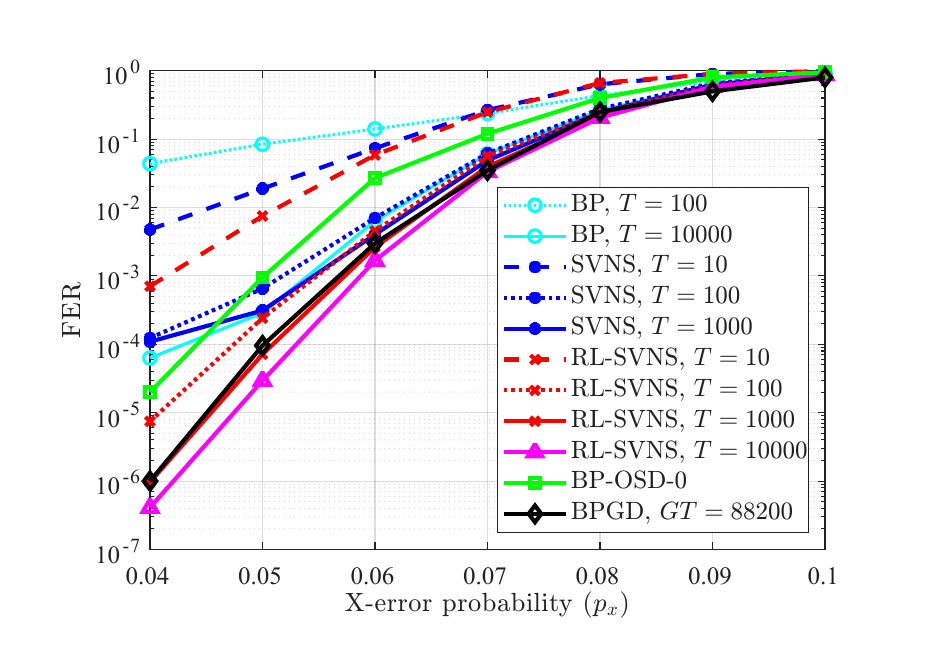}
  \caption{FER for the $[[882, 24, 18\leq d \leq 24]]$ code B1 with flooding BP and SVNS-BP vs. our proposed RL-SVNS-BP decoding algorithm.}
  \label{fig:B1_RL_SVNS}
\end{figure}

Fig.~\ref{fig:B1_RL_SVNS} plots the error-correction performance for our proposed RL trained SVNS (RL-SVNS) decoder along with various other decoding strategies for the B1 code \( [[882,24,18\le d\le 24]] \) \cite{panteleev2021degenerate}. First,  comparing conventional BP flooding decoding (teal circles) and SVNS (randomly selected sequential decoding with no learning, blue circles), we see similar saturating performance as the maximum number of iterations $T$ increases; however, the serial schedule generally converges significantly faster. For example, we observe that beyond roughly $100$ iterations, increasing $T$ yields little gain for SVNS. On the other hand, the RL-SVNS decoder (red X and magenta triangle curves) shows consistently improving performance with increasing $T$, surpassing that of both flooding and SVNS for small $p_x$. We note that, unlike the conventional BP decoders, the RL-SVNS does not appear to indicate an error floor in the simulated range, even with only $T=100$ iterations. 

For reference, the figure also includes the BP-OSD-0 performance (green squares) as well as the limiting performance of BPGD with $T=100$ (black diamond), which results in a global iteration cap of $GT = 88200$. Recall that, in BPGD, after each block of $100$ BP iterations, one decimation is performed, followed by another $100$ BP iterations, and so on, until convergence or until all $n=882$ bits have been decimated. As reported in~\cite{yao2023bpgd}, BPGD with BP using $T=10$ iterations achieves performance close to the $T=100$ variant, indicating that BPGD also exhibits saturation. Finally, we note that for larger $p_x$, BPGD typically requires decimating a large fraction of the bits, which increases the average latency and computational complexity.

\begin{table}[t]
\scriptsize
\centering
\caption{Average number of decoder iterations for standard BP and RL-SVNS on the \( [[882,24,18\le d\le 24]] \) code B1 for different \(X\)-error probabilities \((p_x)\), corresponding to Fig.~\ref{fig:B1_RL_SVNS}.}
\label{tab:avg_iters_b1_bp_rl}
\setlength{\tabcolsep}{5.5pt}
\renewcommand{\arraystretch}{1.15}
\begin{tabular}{|l||c|c|c|c|c|c|c|}
\hline
\multicolumn{1}{|c||}{Method} & \multicolumn{7}{c|}{$X$-error probability \((p_x)\)} \\ \cline{2-8}
\multicolumn{1}{|c||}{} & 0.04 & 0.05 & 0.06 & 0.07 & 0.08 & 0.09 & 0.10 \\ \hline\hline
BP, $T=100$   & 16 & 24.01 & 32.3 & 43.6 & 60.38 & 80.58 & 95.12 \\ \hline
BP, $T=10000$ & 20 & 36.27 & 129 & 729 & 2797 & 6364 & 8993 \\ \hline
RL, $T=100$         &  2.8 &  3.82 &  6.1 & 14.4 & 37.86 & 66.87 & 91.09 \\ \hline
RL, $T=10000$       &  2.8 &  4.19 & 24.9 & 345 & 2125 & 5371 & 8703\\ \hline
\end{tabular}
\end{table}
\begin{table}[t]
\scriptsize
\centering
\caption{Percentage of nonconvergence failures among total decoding errors for BP and RL-SVNS on the \( [[882,24,18\le d\le 24]] \) code B1. Results corresponding to Fig.~\ref{fig:B1_RL_SVNS}.}
\label{tab:nonconv_frac_b1}
\setlength{\tabcolsep}{6pt}
\renewcommand{\arraystretch}{1.15}
\begin{tabular}{|l||c|c|c|c|c|c|c|}
\hline
\multicolumn{1}{|c||}{Method} & \multicolumn{7}{c|}{$X$-error probability \((p_x)\)} \\ \cline{2-8}
\multicolumn{1}{|c||}{} & 0.04 & 0.05 & 0.06 & 0.07 & 0.08 & 0.09 & 0.10 \\ \hline\hline
BP, $T=100$        & 100 & 100 & 100 & 100 & 100 & 100 & 100 \\ \hline
BP, $T=10000$      & 100 & 99  & 100  & 100 & 100 & 100 & 100 \\ \hline
RL, $T=100$   & 95  & 97  & 100 & 99  & 100 & 100 & 100 \\ \hline
RL, $T=10000$ & 29  & 81  & 95  & 99  & 99  & 99  & 99  \\ \hline
\end{tabular}
\end{table}

Table~\ref{tab:avg_iters_b1_bp_rl} reports the average number of decoder iterations per frame for standard BP and RL-SVNS, corresponding to Fig.~\ref{fig:B1_RL_SVNS}. At low \(p_x\), where RL-SVNS achieves over two orders of magnitude lower error rate than standard BP, it also requires substantially fewer iterations on average, indicating both improved reliability and faster convergence in the low-noise regime.\footnote{We note that the complexity per iteration is not equal.  BP and RL-SVNS have exactly the same number of VN updates per iteration, but RL-SVNS has more CN updates since CN $c_j$  will be updated $|\mathcal{N}(c_j)|$ times.}  Table~\ref{tab:nonconv_frac_b1} reports the percentage of decoder failures due to nonconvergence among all decoding errors (\ie, logical errors plus nonconvergence events). As the table shows, RL-SVNS substantially mitigates the nonconvergence issue of the BP decoder for small $p_x$ (where flooding BP displays an error floor).

\subsection{RL-SVNS for Depolarizing Noise}
\begin{figure}[t]
  \centering
  \includegraphics[width=\linewidth]{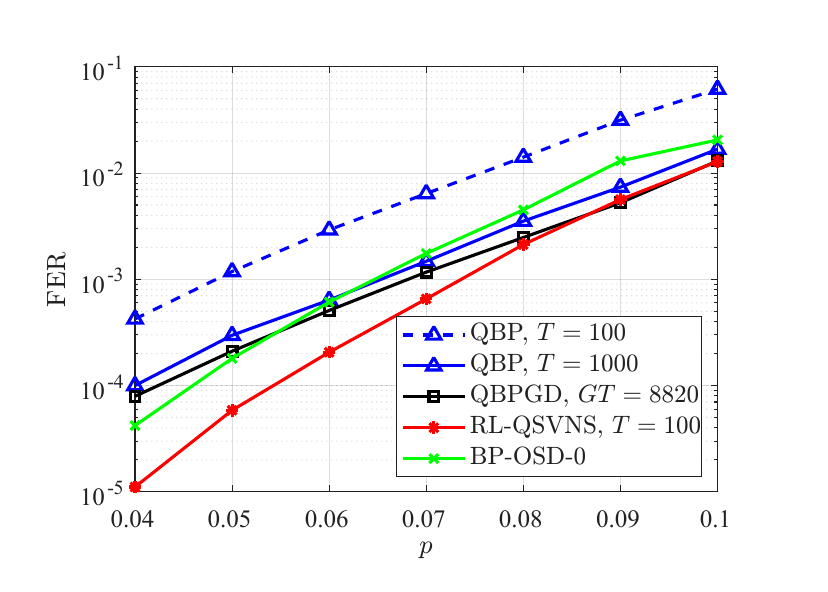}
  \caption{FER for the \( [[882,48,16]] \) code B2 over the depolarizing channel.}
  \label{fig:B2_RL_QSVNS}
\end{figure}

In Fig.~\ref{fig:B2_RL_QSVNS}, we report the FER performance of our proposed reinforcement learning SVNS decoder on the \([[882,48,16]]\) code B2 \cite{panteleev2021degenerate} under depolarizing noise. Across a range of physical error rates \(p\), the RL trained quaternary SVNS decoder (RL-QSVNS) yields more than an order-of-magnitude FER reduction compared to conventional quaternary BP (QBP) when both decoders run for \(100\) iterations. Moreover, RL-QSVNS improves upon the performance of the best possible quaternary BPGD (QBPGD) perforamnce and BP-OSD-0 in the plotted regime. We remind the reader that the RL-QSVNS decoding schedule was trained to jointly optimize both $X$ and $Z$ syndrome mismatches, which we believe accounts for the significant performance improvement observed.

\begin{figure}[t]
  \centering
  \includegraphics[width=\linewidth]{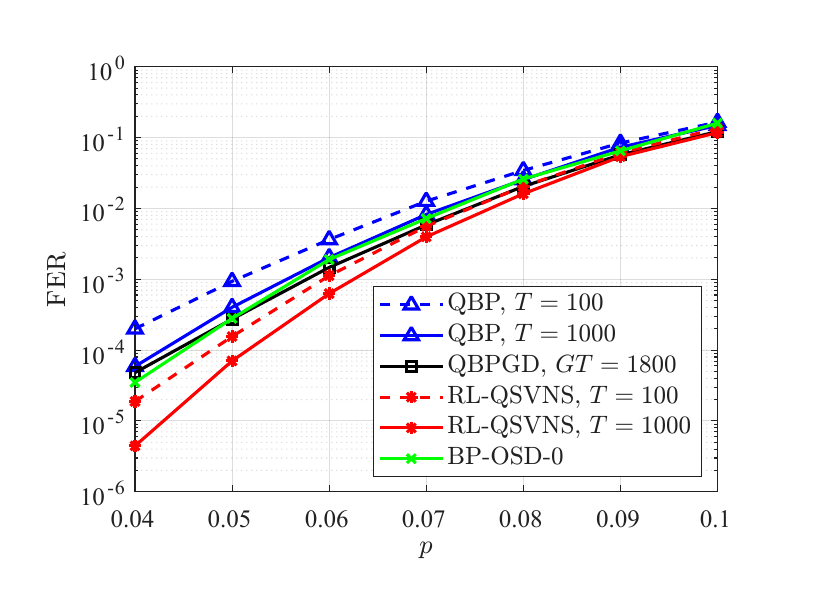}
  \caption{FER for the \( [[180,10,15 \leq d \leq 18]] \) code A5 over the depolarizing channel.}
  \label{fig:A5_RL_QSVNS}
\end{figure}

In Figs.~\ref{fig:A5_RL_QSVNS}-\ref{fig:BB288_QSVNS}, we repeat the same experiment for various QLDPC codes to generally validate the behavior for codes with different structure and parameters. Fig.~\ref{fig:A5_RL_QSVNS} reports the FER performance of our proposed RL-QSVNS decoder for the shorter \( [[180,10,15 \leq d \leq 18]] \) code A5 \cite{panteleev2021degenerate}. At the physical error rate \(p = 0.04\), the learned QSVNS decoder yields about an order-of-magnitude FER reduction compared to QBP when both decoders run for \(100\) iterations. Again, RL-QSVNS improves upon the performance of Q-BPGD and BP-OSD-0 in the plotted regime.

\begin{figure}[t]
  \centering
  \includegraphics[width=\linewidth]{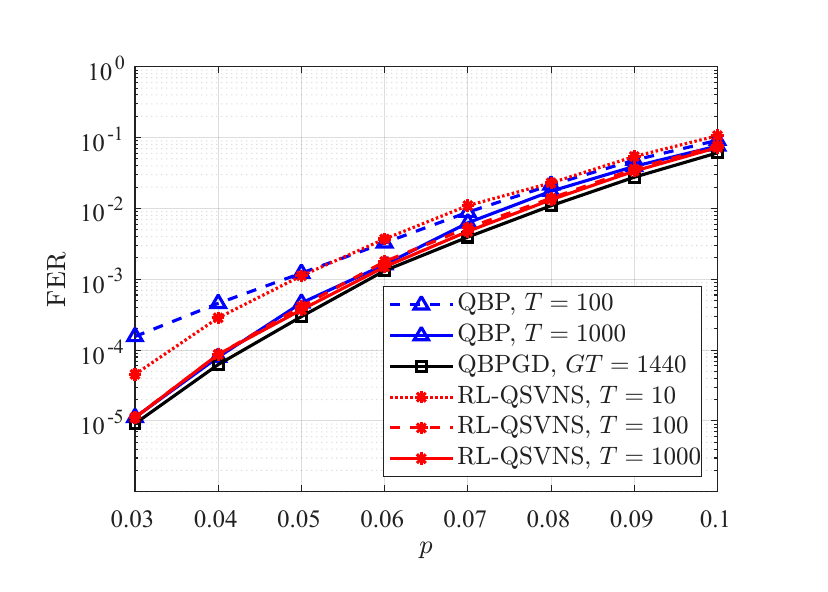}
  \caption{FER for the \( [[144,12,12]] \) BB code over the depolarizing channel.}
  \label{fig:BB144_QSVNS}
\end{figure}

In Fig.~\ref{fig:BB144_QSVNS}, we plot the FER performance of our proposed RL-QSVNS decoder on the bivariate bicycle (BB) \( [[144,12,12]] \) CSS code over the depolarizing channel. For comparison, we also include QBP with several iteration caps, as well as the QBPGD. For a small iteration budget, the RL-QSVNS decoder achieves lower FER than QBP, indicating faster convergence under tight latency constraints, but as the iteration cap increases, the performance of all methods becomes comparable. This suggests that this may be the limiting performance of BP-based methods. Finally, in Fig.~\ref{fig:BB288_QSVNS}, we plot the FER performance of our RL-QSVNS decoder on the longer \( [[288,12,18]] \) BB code for different iteration budgets. Here, we do see a performance advantage of RL (with similar convergence gains). In the same figure, we also report the performance of QBP with different iteration caps, as well as the performance of QBPGD. The results show that our decoder achieves about an order-of-magnitude improvement in FER compared to the other decoders with no indication of the error floor that appears to emerge for the other decoders.

\begin{figure}[t]
  \centering
  \includegraphics[width=\linewidth]{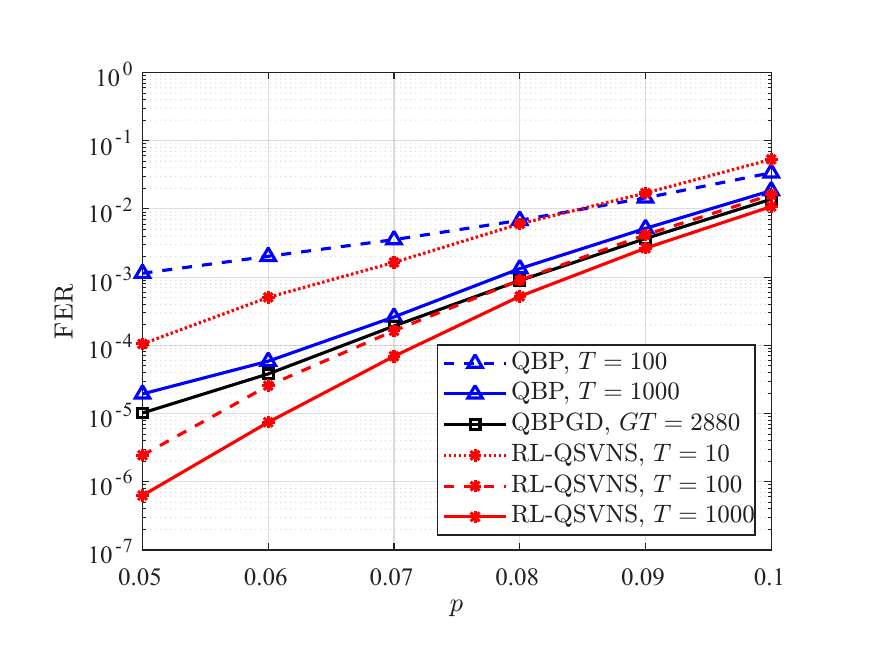}
  \caption{FER for the \( [[288,12,18]] \) BB code over the depolarizing channel.}
  \label{fig:BB288_QSVNS}
\end{figure}

\subsection{RL-QSVNS with Guided Decimation}

\begin{figure}[t]
  \centering
  \includegraphics[width=\linewidth]{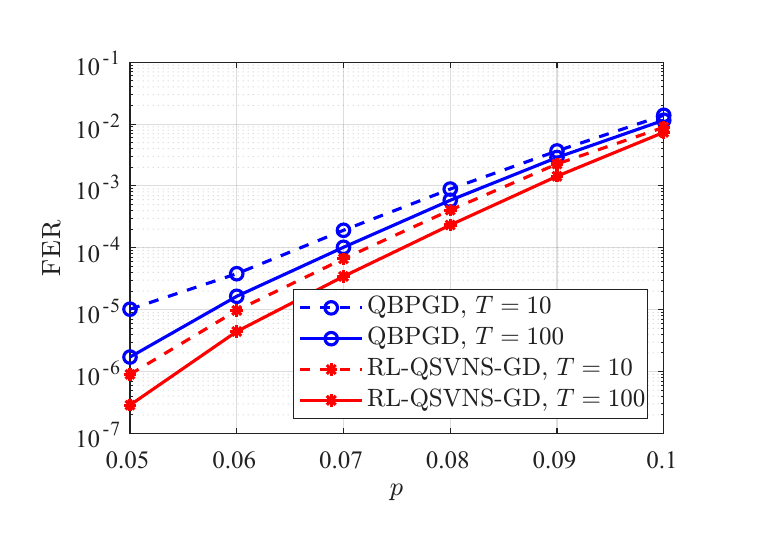}
  \caption{Performance of RL-QSVNS-GD (guided decimation and RL trained sequential decoding) vs. QBPGD for the \( [[288,12,18]] \) BB CSS code over the depolarizing channel.}
  \label{fig:BB288_RL_QSVNS_GD}
\end{figure}
\begin{figure}[t]
  \centering
  \includegraphics[width=\linewidth]{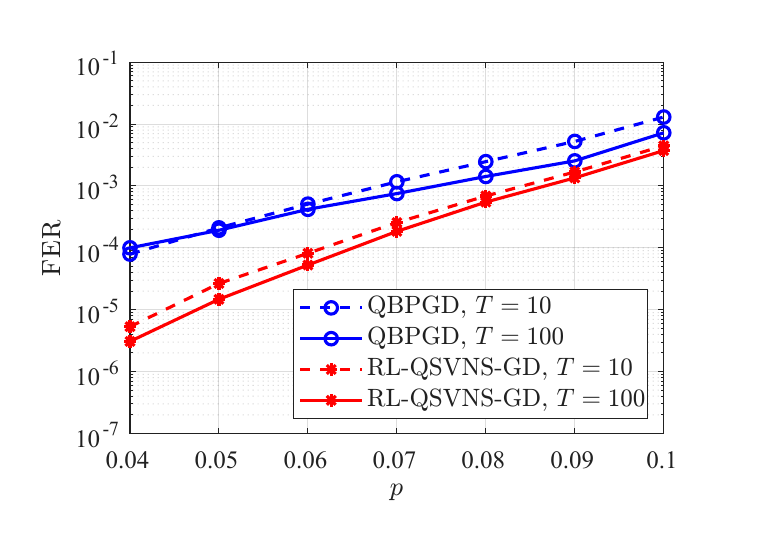}
  \caption{Performance of RL-QSVNS-GD for the \( [[882,48,16]] \) B2 code over the depolarizing channel.}
  \label{fig:B2_RL_QSVNS_GD}
\end{figure}

We also evaluated a hybrid variant in which we replace the flooding BP subroutine inside QBPGD with our RL-SVNS schedule, \ie, no new learning was performed. We refer to this decoder as RL-QSVNS with guided decimation (RL-QVNS-GD). Figs.~\ref{fig:BB288_RL_QSVNS_GD} and \ref{fig:B2_RL_QSVNS_GD} show the resulting FER for the \( [[288,12,18]] \) BB CSS code and the longer \( [[882,48,16]] \) B2 code, respectively, over the depolarizing channel. In both cases, the results reveal that incorporating guided decimation with RL-QSVNS yields a noticeable improvement over standard BPGD across various depolarizing probabilities \(p\).  These results indicate that the learned sequential updates supply more informative intermediate beliefs for the decimation mechanism.

Tables~\ref{tab:avg_decimations_bb288} and \ref{tab:avg_decimations_b2}  report the average number of decimation steps per decoding attempt for the two methods corresponding to Figs.~\ref{fig:BB288_RL_QSVNS_GD} and \ref{fig:B2_RL_QSVNS_GD}, respectively. The results show that replacing flooding (QBPGD) with learned sequential scheduling (RL-QSVNS-GD) substantially reduces the number of required decimation steps across all tested values of \(p\), with the gap widening as the channel becomes noisier. This reduction  suggests that RL-QSVNS delivers more reliable soft information during the guided-decimation process, thereby decreasing the amount of decimation needed to reach a valid correction.

\begin{table}[t]
\scriptsize
\centering
\caption{Average number of decimation steps per decoding attempt for standard QBPGD and RL-QSVNS-GD for the \( [[288,12,18]] \) BB CSS code over the depolarizing channel (corresponding to Fig.~\ref{fig:BB288_RL_QSVNS_GD}).}
\label{tab:avg_decimations_bb288}
\setlength{\tabcolsep}{6pt}
\renewcommand{\arraystretch}{1.15}
\begin{tabular}{|l||c|c|c|c|c|c|}
\hline
\multicolumn{1}{|c||}{Method} & \multicolumn{6}{c|}{Depolarizing probability \(p\)} \\ \cline{2-7}
\multicolumn{1}{|c||}{} & 0.05 & 0.06 & 0.07 & 0.08 & 0.09 & 0.10 \\ \hline\hline
BPGD, $T=10$             & 1.09 & 1.16 & 1.32 & 1.73 & 2.95 & 6.77 \\ \hline
BPGD, $T=100$             & 1.01 & 1.01 & 1.05 & 1.24 & 2.09 & 5.12 \\ \hline
RL, $T=10$      & 1.00 & 1.00 & 1.03 & 1.15 & 1.77 & 3.95 \\ \hline
RL, $T=100$      & 1.00 & 1.00 & 1.01 & 1.07 & 1.50 & 3.50 \\ \hline
\end{tabular}
\end{table}


\begin{table}[t]
\scriptsize
\centering
\caption{Average number of decimation steps per decoding attempt for standard QBPGD and RL-QSVNS-GD for the \( [[882,48,16]] \) code B2 over the depolarizing channel (corresponding to Fig.~\ref{fig:B2_RL_QSVNS_GD}).}
\label{tab:avg_decimations_b2}
\setlength{\tabcolsep}{6pt}
\renewcommand{\arraystretch}{1.15}
\begin{tabular}{|l||c|c|c|c|c|c|c|}
\hline
\multicolumn{1}{|c||}{Method} & \multicolumn{7}{c|}{Depolarizing probability \(p\)} \\ \cline{2-8}
\multicolumn{1}{|c||}{} & 0.04 & 0.05 & 0.06 & 0.07 & 0.08 & 0.09 & 0.10 \\ \hline\hline
BPGD, $T=10$        & 1.11 & 1.29 & 1.66 & 2.54 & 4.24 & 7.85 & 16.53 \\ \hline
BPGD $T=100$        & 1.09 & 1.18 & 1.40 & 1.72 & 2.44 & 3.80 & 8.09 \\ \hline
RL, $T=10$ & 1.00 & 1.02 & 1.05 & 1.14 & 1.42 & 2.15 & 4.54  \\ \hline
RL, $T=100$ & 1.00 & 1.00 & 1.01 & 1.02 & 1.15 & 1.55 & 2.68  \\ \hline
\end{tabular}
\end{table}

\section{Conclusions}\label{sec:conclusion}

In this paper, we proposed a novel decoding approach for QLDPC codes that learns the VN update order for sequential BP, yielding an RL trained sequential VN scheduling (RL-SVNS) decoder. To make RL-SVNS more efficient, we proposed a fast implementation that computes the RL states and update rewards using local, incrementally maintained quantities, and performs greedy action selection efficiently using a heap-based priority structure. Numerical simulation results show that the RL-SVNS decoder provides faster decoder convergence and substantial FER improvements over both flooding BP and random sequential scheduling for a selection of representative QLDPC codes. In addition, its performance is often substantially better than BP-OSD and BPGD decoding while maintaining a complexity comparable to standard BP. 

We also demonstrated that the approach was modular, \ie, the proposed RL-SVNS schedule can be combined with other techniques, providing a further way to improve both convergence behavior and end-to-end decoding performance. We demonstrated this by combining 
quaternary BPGD with RL-SVNS. 
The resulting hybrid decoder (RL-QSVNS-GD) achieves a noticeable performance gain over standard QBPGD under the same global iteration cap, while also requiring significantly fewer decimation steps on average. 

\section{Acknowledgment}
This material is based upon work supported by the National Science Foundation under Grant No. CCF-2145917.

\bibliographystyle{IEEEtran}

\bibliography{bibliography}

\end{document}